\documentclass[12pt,a4paper]{article}
\usepackage[T1]{fontenc}
\usepackage[utf8]{inputenc}
\usepackage{mathtools,amsmath,amsthm,array,graphicx,natbib,color,xspace,enumitem,multirow,tikz}
\usepackage[textsize=tiny]{todonotes}
\usepackage{caption}
\usepackage{pgfplots}
\usetikzlibrary{decorations.pathreplacing}
\usetikzlibrary{math}

\newtheorem{mydef}{Definition}
\newtheorem{prop}{Proposition}

\newtheorem{example}{Example}
\newtheorem{lemma}{Lemma}

\usepackage{fancyhdr} 
\usepackage{todonotes}

\def\cdf(#1)(#2)(#3){0.5*(1+(erf((#1-#2)/(#3*sqrt(2)))))}%

\begin{document}


\begin{center}\huge{Observing Actions in Global Games}
\end{center}

\begin{center} \textit{Dominik Grafenhofer, Wolfgang Kuhle}\footnote{Wolfgang Kuhle (corresponding author), Zhejiang University, Hangzhou, China and MEA, Max Planck Institute for Social Law and Social Policy, Munich, Germany and VSE, Prague, Czech Republic, E-mail wkuhle@gmx.de. Dominik Grafenhofer, Deutsche Telekom, Bonn, Germany, E-mail econ@grafenhofer.at. Most of this paper was written during our time at the Max Planck Institute in Bonn, and we thank Martin Hellwig and Carl Christian von Weizsäcker for helpful and encouraging conversations. We also thank Philipp Koenig, Stephen Morris, Ruqu Wang, Wenzhang Zhang and Seminar participants in Bonn, Hangzhou, Prague, Tübingen, and the 30th international game theory conference in Stony Brook for comments and questions. The views and opinions expressed in this paper are solely those of the authors and do not necessarily represent the views of Deutsche Telekom.}
\end{center}


\noindent\emph{\textbf{Abstract:} We study Bayesian coordination games where agents receive noisy private information over the game's payoffs, and over each others' actions. If private information over actions is of low quality, equilibrium uniqueness obtains in a manner similar to a global games setting. On the contrary, if private information over actions (and thus over the game's payoff coefficient) is precise, agents can coordinate on multiple equilibria. We argue that our results apply to phenomena such as bank-runs, currency crises, recessions, or riots and revolutions, where agents monitor each other closely.}

\hspace{2.5cm}

\noindent\emph{ERNIE: There is something funny going on over there at the
bank, George, I've never really seen one, but that's got all the
earmarks of a run.}\\
\emph{PASSERBY: Hey, Ernie, if you have any money in the bank, you
better hurry.}\footnote{From the movie script "It's a wonderful life."}


\section{Introduction} 

The global games approach assumes that players face asymmetric private information over the game's payoff structure. In turn, each agent has to use his own signal over the game's payoffs to infer the signals, and thus the actions, of the other agents. Such inference makes it difficult to coordinate on multiple equilibria, and whenever private information over the game's payoffs is very precise, but not perfect, the global games structure selects unique equilibria.

One implicit assumption in the global games approach of \citet{Rub89}, \citet{Car93} and \citet{Mor98} is that agents cannot observe each others' actions directly. At the same time, global games are used extensively to model phenomena such as bank-runs, riots and revolutions, or joint investment projects, where agents study each others' actions. That is, in the context of a bank-run, depositors can observe the length of a queue, respectively the lack thereof, in front of their local bank branch. To understand the equilibria that agents play, we therefore argue that agents' information over the game's payoff coefficients \textit{and} information over each others' actions should be taken into account. 

Information over actions plays a dual role in the current model. First, actions depend on the game's payoff coefficient, and signals over actions therefore carry information over the game's payoff coefficients. In this interpretation, the signal over actions is just another signal over the game's fundamental. Taking this view, private information over actions should reinforce the global games mechanism, where private information over fundamentals selects unique equilibria. The second function of signals over actions is that they inform players of each others' actions, which helps coordination. In equilibrium, we find that this second effect dominates, and multiple equilibria are ensured whenever private information over actions is sufficiently precise.

Multiple equilibria obtain in a manner similar to an epidemic. That is, in the context of a bank-run, increases in the mass of agents who line up to withdraw their deposits are observed, and induce additional agents to join the queue. This infectious process allows agents to coordinate on multiple equilibria whenever the precision, with which they observe each other, is sufficiently high. The equilibria of the game are then situations where no agent, given his information over the queue and the game's fundamental, would like to join/leave the queue of agents, who are waiting for the bank to open. In turn, once the bank opens, the agents in the queue withdraw their money.

Technically, equilibria are steady states, where the mass of waiting agents in period $n-1$, $A^{n-1}$, which is partially revealed through agents' private information, coincides with the mass of waiting agents $A^n$, such that $A^n=A^{n-1}=A$. In the main text we will suppress the time index $n$, and solve for steady states only. We do use the time index in appendices \ref{Iteration} and \ref{proof_infogap_mult_1d} to obtain these steady states via a sequential model, where agents join and leave the crowd of agents, who are waiting to withdraw their money, until a steady state is reached.

\emph{Alternative Interpretations:} Riots and revolutions, joint investment projects, or coordinated attacks are problems where actions are strategic complements. Agents will therefore monitor each other closely.  

During riots and revolutions, people see whether the number of protesters in the street is large or small. Bystanders, who observe that turn-out is large, may decide to join the protest. Protesters, who see that turn-out is small, may withdraw. The present model studies the steady state crowd, which gathers \textit{before} the protesters clash with the riot police.\footnote{The January 2021 storming of the US Capitol may serve as an example for this: protesters were gathering at various rallies throughout Washington before they moved to overwhelm the police.} 

Political parties and parliaments often perform nonbinding test-votes.\footnote{One may view such votes as public signals. In an alternative interpretation, one may consider the relevant information as being private in the sense that the voting population of politicians discusses the composition of the aggregate 'yes' and 'no' votes in small circles. That is, politicians discuss in private \textit{who} voted 'yes' and 'no', to understand the outcome of the final, binding, vote.} These test votes can be interpreted as coordination devices, which allow politicians to observe the "queue" of politicians who support a particular motion. After these test-votes, they proceed to the final, binding, vote. 

In the context of large investment projects, investors sign nonbinding letters of intent. These letters serve as signals regarding the future action, i.e., the expected number of investors who will actually participate once the binding investment contract is signed.

\emph{Related literature:} Agents who observe each others' actions learn about the game's fundamental. This learning channel has been studied by \citet{Cha99}, \citet{Pav07,Pav07b}, \citet{Fra12}, \citet{Loe14}. These models assume that agents observe (i) the actions that they took when they played the game in the past, or (ii) the \textit{irreversible}\footnote{That is, agents are partitioned into early and late movers; late movers observing the irreversible actions taken by those who moved early.}\footnote{See \citet{Kov13} for a two period global games model with reversible actions. \citet{Kov13} focus on environments where actions are unobservable, and thus equilibria are unique.} actions of a group of early movers. This type of information over actions is equivalent to information over the game's fundamental. The global games mechanism thus applies, and ensures unique equilibria whenever private information over actions/fundamentals is sufficiently precise. 

The current paper, on the contrary, studies steady states of crowds that form \emph{before} the final, binding, action is taken. In the context of a bank-run, an equilibrium is a steady state queue of agents that forms before the bank opens. That is, an equilibrium is a situation where none of the agents in the queue, given their information over the queue's length and their information over the bank's reserves, would like to withdraw from the queue. At the same time, none of the agents outside of the queue, given their information over the queue's length and their information over the bank's reserves, would like to join the queue.\footnote{In the alternative context of a riot, we study the crowd that gathers \emph{before the actual attack takes place}. On the contrary, \citet{Cha99}, \citet{Pav07,Pav07b}, \citet{Fra12}, \citet{Loe14}, study agents who learn from the outcomes of past attempts to, e.g. storm the US capitol, or from past actions of early movers who have already committed to the attack.}

\citet{Ang06} find that public signals over players' steady state actions, just like public signals over fundamentals, help agents to coordinate on multiple equilibria. The present model shows that private information over actions, unlike private information over fundamentals, induce multiple equilibria.

Modern technology increasingly allows agents to draw on diverse sources of information. Taking this perspective, we contribute to the effort aimed at enriching the global games information structure. \citet{Izm10}, \citet{Ste11} \citet{Kuh15}, \citet{Gra16}, \citet{Ber16}, \citet{Bin01} have recently emphasized the role of heterogenous priors, strategic information revelation and learning, environments where agents receive signals over each others' information. 


Finally, the current model connects the literature on conjectural equilibrium, \citet{Bat97}, \citet{Min03}, \citet{Rub94}, \citet{Esp13},\footnote{See also \citet{Hah77} for Walrasian economies, where agents hold conjectures over each others' supply and demand functions, which need not be true.} emphasizing noisy information over steady state actions, with arguments from the literature on global games, where uncertainty over actions originates from parameter uncertainty.  

\emph{Organization:} Section \ref{model} outlines the model. Section \ref{MorrisShin} recalls the \citet{Mor04} model. Section \ref{Normal} provides a tractable example of global games featuring private signals over actions. In Section \ref{information_equilibrium} we prove our results for more general signal structures. Section \ref{Discussion} concludes.
 

\section{Model} \label{model}

There is a status quo and a unit measure of agents indexed by
$i\in[0,1]$. Each of these agents $i$ can choose between two
actions $a_i\in\{0,1\}$. Choosing $a_i=1$ means to attack the
status quo. Choosing $a_i=0$ means that the agent abstains from
attacking the status quo. An attack on the status quo is
associated with a cost $c\in(0,1)$. If the attack is successful,
the status quo is abandoned, and attacking agents receive a net
payoff $1-c>0$. If the attack is not successful, an attacking
agent's net payoff is $-c$. The payoff for an agent who does not
attack is normalized to zero. The status quo is abandoned if the
aggregate size of the attack $A:=\int_0^1a_idi$ exceeds the
strength of the status quo $\theta$, i.e., if $A>\theta$.
Otherwise, if $A<\theta$, the status quo is maintained, and the
attack fails.

The timing of the game may be thought of as follows:

\begin{itemize}

\item A group of agents, who observe each other, gathers in front of the bank. An equilibrium is a situation where no agent, given his information over the length of the queue of waiting agents and the game's fundamental, would like to join/leave the queue.

\item The bank opens, and the agents, who are waiting in the queue, withdraw their money.

\end{itemize}

In what follows, we study steady state equilibria, where the past queue of wating agents $A^{n-1}$, which agents observe, coincides with the contemporaneous queue of wating agents $A^n$ such that $A^n=A^{n-1}=A$. That is, in the main text we suppress the time index and solve for steady states only. We do use the time index $n$ in appendices \ref{Iteration} and \ref{proof_infogap_mult_1d} to obtain these steady states via an iteration over the crowd of attacking agents.
Finally, in order to focus on the role of private information, we assume throughout the paper that players hold a uniform uninformative prior over $\theta$. The game, with the exception of the exogenous fundamental $\theta$ and the endogenous size of the attack $A(\theta)$, is common knowledge.

\subsection{The \citet{Mor04} Benchmark: Observing Actions Indirectly}\label{MorrisShin}
Let us briefly recall the \citet{Mor04} benchmark model, where agents receive private information over the game's fundamental only. That is,
agents receive signals
\begin{eqnarray}  x_i=\theta+\sigma_x\epsilon_i, \quad \epsilon_i\sim\mathcal{N}(0,1),  \end{eqnarray}  
which inform them of the game's fundamental $\theta$ with precision $\alpha_x:=\frac{1}{\sigma_x^2}$. Expected utility is thus:
\begin{eqnarray}  E[U(a_i)|x_i]=a_i(P(\theta<A|x_i)-c). \label{MS2} \end{eqnarray}
We denote by $x^*$ the signal threshold that separates signal values $x_i<x^*$, for which agents attack $a_i=1$, from signal values $x_i>x^*$, for which agents do not attack $a_i=0$. Using (\ref{MS2}), we have an indifference condition 
\begin{eqnarray}  P(\theta<A|x^*)=c,  \label{MS3}\end{eqnarray}
which characterizes $x^*$.
Given threshold $x^*$, and given a fundamental $\theta$, the mass of attacking agents is
\begin{eqnarray}  A=P(x<x^*|\theta)=\Phi(\sqrt{\alpha_x}(x^*-\theta)),  \label{MS4}\end{eqnarray} 
where $\Phi()$ denotes the cumulative density function of the standard normal distribution.
Equation (\ref{MS4}), helps to compute threshold values $\theta^*$, which separate fundamental values $\theta<\theta^*$, for which the attack succeeds from fundamental values $\theta>\theta^*$, for which the attack fails:
\begin{eqnarray}  A(\theta^*)=\theta^* \quad \Leftrightarrow \quad \Phi(\sqrt{\alpha_x}(x^*-\theta^*))=\theta^*.  \label{MS5}\end{eqnarray}
Combining (\ref{MS3}) and (\ref{MS5}) we rewrite the payoff indifference condition:
\begin{eqnarray}  P(\theta<\theta^*|x^*)=\Phi(\sqrt{\alpha_x}(\theta^*-x^*))=c.   \label{MS6}\end{eqnarray}
Together (\ref{MS5}) and (\ref{MS6}) yield:
\begin{prop} There exists a unique equilibrium $x^*,\theta^*$. \label{Prop0}\end{prop}
\begin{proof} We solve equation (\ref{MS6}) for $x^*=-\Phi^{-1}(c)\frac{1}{\sqrt{\alpha_x}}+\theta^*.$ Substituting into (\ref{MS5}) yields $\theta^*=\Phi(-\Phi^{-1}(c))$. In turn, $x^*=-\Phi^{-1}(c)\frac{1}{\sqrt{\alpha_x}}+\Phi(-\Phi^{-1}(c)).$ \end{proof}
That is, equilibria are unique regardless of the precision with which agents observe the game's fundamental.

\textit{How well do agents observe each others' actions?} Agents receiving signals $x_i$ over the game's fundamental $\theta$ know that the other agents' equilibrium action is an invertible function of the fundamental $A=\psi(\theta)$ such that $\theta=\psi^{-1}(A)$. Hence, the signal over the game's fundamental informs agents of each others' contemporaneous actions. Rewriting (\ref{MS4}), we have $\theta=-\frac{1}{\sqrt{\alpha_x}}\Phi^{-1}(A)+x^*$. Signals $x_i=\theta+\sigma_x\epsilon_i$ thus carry information over actions $x_i=-\frac{1}{\sqrt{\alpha_x}}\Phi^{-1}(A)+x^*+\sigma_x\epsilon_i$, and we have $(x^*-x_i)\sqrt{\alpha_x}=\Phi^{-1}(A)-\epsilon_i$. That is, information is specified such that the precision with which agents observe actions is independent of the precision $\alpha_x$ with which they observe the game's fundamental.

\subsection{Observing Actions Directly}\label{Normal}
We begin with a simple example, in which agents receive only one private signal over the attack's net size
$A-\theta$, rather than two separate signals over $\theta$ and $A$. In Section \ref{information_equilibrium}, we extend the current result to more general signal structures and distribution functions. Moreover, appendices \ref{Iteration} and \ref{Tipping} derive the (steady state) equilibrium function $A(\theta)$ via an iteration argument in which agents observe past actions.

The signal
\begin{eqnarray}z_i=A-\theta+\sigma_z\epsilon_i, \quad \epsilon_i\sim\mathcal{N}(0,1) \label{i1}\end{eqnarray}
informs players $i$ with precision $\alpha_z:=\frac{1}{\sigma^2_z}$ of the attack's net size
$A-\theta$, and we have:
\begin{prop}
If private information is precise, $\alpha_z>2\pi$, agents
can coordinate on multiple equilibria.\end{prop}\label{Normalprop}
\begin{proof}
We proceed in three steps. First, we compute the threshold signal
$z^*$, for which agents are indifferent between attacking and not
attacking. Given this threshold, we compute the mass of attacking
agents. Third, we show that there exist multiple equilibria.

1.) Payoff indifference condition: Given a signal $z_i$,
agents $i$ choose an action $a_i\in\{0,1\}$, to maximize expected
utility:
\begin{eqnarray} E[U_i]=a_i(P(A-\theta>0|z_i)-c). \end{eqnarray}
Agent $i$ is therefore just indifferent between attacking,
$a_i=1$, and not attacking, $a_i=0$, when he receives a signal
$z_i=z^*$ such that:
\begin{eqnarray}  P(A-\theta>0|z^*)=c. \label{l0.1} \end{eqnarray}

It follows from (\ref{l0.1}) that agents attack if $z>z^*$, and
they will abstain from attacking whenever $z\leq z^*$.

2.) Given the critical signal $z^*$, we can compute the mass of
attacking agents:
\begin{eqnarray} A=P(z>z^*|A,\theta)  \label{A1}\end{eqnarray}
For normally distributed signal errors, (\ref{A1}) can be
rewritten as:
\begin{eqnarray} A=1-\Phi(\sqrt{\alpha_z}(z^*-A+\theta)), \label{A2}  \end{eqnarray}
where $\Phi()$ is the cumulative normal distribution. From
(\ref{A2}), we have 
\begin{lemma} \label{lem1} If $\alpha_z>2\pi$ then, for every level $z^*$, there exists an
interval $[\check{\theta}(z^*),\hat{\theta}(z^*)]$ such that
(\ref{A2}) has three solutions $A_j(\theta,z^*), j=1,2,3$ whenever
$\theta\in[\check{\theta}(z^*),\hat{\theta}(z^*)]$.
\end{lemma}

To construct equilibrium functions $A(\theta,z^*)$, we use the
solutions $A_j(\theta;z^*)$ from Lemma \ref{lem1}. More
specifically, we focus on the solutions $j=1,3$ such that we
obtain functions $A_t(\theta;z^*)$, which are downward sloping
$\frac{\partial A_t}{\partial \theta}<0$. Given these functions,
it remains to show that there exist values $z^*$ that satisfy the
payoff indifference condition
\begin{eqnarray}  P(A_t(\theta,z^*)-\theta>0|z^*)=c. \label{l0}
\end{eqnarray} In Appendix \ref{A3} we show that there exist small values $\check{z}$ such that
$P(A_t(\theta,\check{z})-\theta>0|\check{z})>c$ and large values
$\hat{z}$ such that $P(A_t(\theta,\hat{z})-\theta>0|\hat{z})<c$.

Equilibrium values $z^*$ now either obtain as solutions to (\ref{l0}), or the critical $z^*$
values are those values where the function
$P(A_t(\theta,z^*)-\theta>0|z^*)$ is discontinuous in $z^*$. In
that case we have a $z^*$, such that for small $\delta>0$,
$P(A_t(\theta,z^*+\delta)-\theta>0|z^*+\delta)>c$ and
$P(A_t(\theta,z^*-\delta)-\theta>0|z^*-\delta)<c$. That is, the
expected value of attacking/not attacking changes discontinuously
at the agents' equilibrium cutoff value $z^*$.
\end{proof}

Intuitively, multiple equilibria obtain since increases in the mass of attacking agents are observed. This, in turn, increases the number of attacking agents, which is again visible and induces even more agents to run... . Hence, if the private signal's precision is sufficiently high, runs feed on themselves, and agents can coordinate on multiple equilibria.



\section{Separate Signals over Actions and Fundamentals}\label{information_equilibrium}

Let us now assume that each player $i$ receives two pieces of private information: a noisy signal $x_i$ over the strength of the status quo $\theta$, and another signal $y_i$ over the other players' (steady state) actions $A$:
\begin{equation} \label{def_signals_x_and_y}
\begin{array}{lrcl}
\text{Signal over the fundamental: } \qquad\qquad & x_i &=& \theta + \epsilon^x_i \\
\text{Signal over the aggregate attack: } \qquad\qquad &  y_i &=& A + \epsilon^y_i
\end{array}
\end{equation}

Agents choose actions $a_i$ to maximize: 
\begin{eqnarray} E[U(a_i)|x_i,y_i]=a_i(P(\theta<A|x_i,y_i)-c). \end{eqnarray}
Action $a_i=1$ is thus optimal whenever $P(\theta<A|x_i,y_i)\geq c$. Agents are just indifferent between attacking and not attacking when signals $x_i$ and $y_i$ are such that $P^*:=c$. We denote the error terms' joint density function by $f(\epsilon^x,\epsilon^y)$. Given agents' information, and the critical probability $P^*$, we can define:  
\begin{mydef}[Equilibrium] \label{fe}
An aggregate attack function $A(\theta)$ is an equilibrium of the game, if for all $\theta\in\mathbf{R}$ the following holds:\footnote{$\chi$ denotes the indicator function.}
\begin{equation}\label{eq_consistency}
A(\theta) = \int_{\mathbf{R}^2} \chi_{\{(\epsilon^x,\epsilon^y):P[A(\theta)\geq\theta|x_i=\theta+\epsilon^x,y_i=A(\theta)+\epsilon^y,A(\cdot)]\geq P^*\}} \, f(\epsilon^x,\epsilon^y) \, d\epsilon^x d\epsilon^y.
\end{equation}
\end{mydef}

To characterize equilibria, we treat cases where error terms are bounded or unbounded separately:
\begin{prop}\label{cs_mult}
Suppose that $\epsilon^y_i\in[-\sigma,\sigma]$, i.e. $f(\epsilon^x,\epsilon^y)=0$ for all $|\epsilon^y|>\sigma$ and all $\epsilon^x\in\mathbf{R}$. Further assume that the precision of the signal about the aggregate attack $y_i$ is precise enough, i.e. $0<\sigma<\frac12$. Then, there exists a continuum of equilibria.
\end{prop}
\begin{proof} See Appendix \ref{A222} \end{proof}

\begin{prop}\label{infogap_mult_general}
Assume that $\epsilon^x_i$ and $\epsilon^y_i$ are distributed according to pdf $f=f_x f_y$ (cdfs $F_x$, $F_y$), where $f_x$ and $f_y$ are symmetric.\footnote{The symmetry assumption shortens the proof.} There exist $\delta> 0, \gamma> 0,$ and $\xi > 0$ such that $1-\delta \geq \gamma$, $1>3\delta + 2\gamma$ and the following conditions hold:
\begin{align}
&\frac{F_x(\xi)}{1-F_x(\xi)} \, \frac{\sup_{a\in[0,\delta]}f_y(\eta-a)}{\inf_{a\in [1-\delta,1]}f_y(\eta-a)} \leq \frac{1-c}{c} \quad \text{ for all }\eta\geq 1-\delta-\gamma \label{cond1_2d}\\
&F_x(\xi)F_y(\gamma)\geq 1- \delta \label{cond2_2d} 
\end{align}
Whenever these conditions hold, there exists a continuum of equilibria.
\end{prop}
\begin{proof} See Appendix \ref{proof_infogap_mult_general}.\end{proof}
The proof of Proposition \ref{infogap_mult_general} relies on an iteration argument: we start with a guess $A^0_t(\theta)$ and compute a best response $A^1_t(\theta)$. In turn, $A^1_t(\theta)$ yields another best response $A^2_t(\theta)$ and so on... . Via this iteration we construct a converging sequence of aggregate attacks/best responses. In the limit, we obtain an equilibrium, steady state, attack function for every $t$. Finally, since there is a continuum of permissible $t$ values, we have a continuum of equilibria. In particular, if error terms are normally distributed, we have 


\begin{example}\label{example}
Suppose error terms $\epsilon^x,\epsilon^y$ are normally distributed, and denote the precisions of the respective signals by $\alpha_x=\frac{1}{\sigma^2_x}$ and $\alpha_y=\frac{1}{\sigma^2_y}$. For that case, we have $f_x(\xi)=\phi(\sqrt{\alpha_x}\xi)$ and $F_x(\xi)=\Phi(\sqrt{\alpha_x}\xi)$ and $f_y(\gamma)=\phi(\sqrt{\alpha_y}\gamma)$ and $F_y(\gamma)=\Phi(\sqrt{\alpha_y}\gamma)$, where $\phi$ and $\Phi$ represent the density and cumulative density functions of the standard normal distribution. In turn, we can choose $\delta$ and $\gamma$ such that $1>3\delta+2\gamma$, e.g., $\delta=.2$ $\gamma=.1$. Moreover, we can choose $\alpha_x\xi>\Phi^{-1}(.8)$. Finally we let $\alpha_y\rightarrow\infty$, such that (\ref{cond1_2d}) and (\ref{cond2_2d}) are both satisfied.\end{example}

\section{Discussion}\label{Discussion} 

The basic global games approach to bank-runs assumes that agents cannot observe the queues that form before the bank opens. The only information that agents posses concerns the bank's financial strength. In turn, each agent has to use his own signal over the bank's strength to infer the signals, and thus the actions, of the other agents. This indirect reasoning makes coordination difficult. 

The current model emphasizes that bank-runs, riots and revolutions, or joint investment projects, are phenomena where agents monitor each other closely. That is, queues in front of a bank are observable, and induce additional depositors to withdraw their money. This infectious process ensures multiple equilibria whenever the precision with which agents observe each other is sufficiently high. Equilibria in our model are then situations where no agent, given his information over the length of the queue of waiting agents and the game's fundamental, would like to join/leave the queue.

Our model's comparative statics accommodate a range of environments that differ regarding the information that agents have over each others' actions and over fundamentals. Equilibrium uniqueness obtains in a manner similar to a global games setting whenever information over actions is of low quality. On the contrary, if private information over actions (and thus the game's fundamental) is precise, agents can coordinate on multiple equilibria. These equilibria are quite different from what basic global games predict. Hence, we argue that the current model helps to better understand the equilibria that agents play in environments where it is reasonable to assume that they can observe each other.



\newpage
\appendix


\section{Conditional Probabilities}\label{A3}
We have to show that, $\lim_{z^*\rightarrow -\infty}
P(A(\theta,z^*)-\theta>0|z^*)=0$ and $\lim_{z^*\rightarrow \infty}
P(A(\theta,z^*)-\theta>0|z^*)=1$. To do so, we recall that
$z_i=A-\theta+\sigma_z\epsilon_i$ with $\epsilon_i\sim\mathcal{N}(0,1)$, and
examine the conditional probability:
\begin{eqnarray}  P(A(\theta,z^*)-\theta>0|z^*)=c. \label{l01} \end{eqnarray}
We begin by defining $y(\theta,z^*):=A(\theta,z^*)-\theta$, and we
recall Bayes's formula:
\begin{eqnarray}  f(y|z^*)=\frac{h(z^*|y)f(y)}{\int_{-\infty}^{\infty} h(z^*|y)f(y)dy}, \label{l02} \end{eqnarray}
Where $h(z^*|y)$ is a normal distribution. Moreover, we note that:
\begin{eqnarray} f(y(\theta))=g(\theta(y))\frac{d \theta}{d y}. \label{l06} \end{eqnarray}
Where $\frac{d\theta}{dy}=\frac{1}{A_{\theta}(\theta,z^*)-1}$.
Recalling (\ref{A2}), $A=1-\Phi(\sqrt{\alpha_z}(z^*-A+\theta))$,
we have
$A_{\theta}(\theta,z^*)=\frac{-\sqrt{\alpha_z}\phi(\sqrt{\alpha_z}(z^*-A+\theta))}{1-\sqrt{\alpha_z}\phi(\sqrt{\alpha_z}(z^*-A+\theta))}<0$,
for $A_j, j=1,3$. Agents hold a uniform uninformative prior over
$\theta$, such that $g(\theta)$ is a constant $\bar{g}$. Moreover,
we have $\lim_{z^*\rightarrow \infty}A_{\theta}=0$ and thus
$\lim_{z^*\rightarrow \infty}\frac{d\theta}{dy}=-1$ and
$f(y)=-1\bar{g}$. Substituting into (\ref{l02}) yields:
\begin{eqnarray}  \lim_{z^*\rightarrow \infty}f(y|z^*)=\lim_{z^*\rightarrow \infty}\frac{h(z^*|y)}{\int_{-\infty}^{\infty} h(z^*|y)dy}, \label{l02a} \end{eqnarray}
where $h(z^*|y)=\phi(\sqrt{\alpha_z}(z^*-y))$. Finally, we have:
\begin{eqnarray}  \lim_{z^*\rightarrow \infty}P(y>0|z^*)=\lim_{z^*\rightarrow \infty}\frac{\int_{0}^{\infty}h(z^*|y)}{\int_{-\infty}^{\infty} h(z^*|y)dy}=\lim_{z^*\rightarrow \infty}\Phi(\alpha_z z^*)=1. \label{l02b} \end{eqnarray}
The same argument can be made to show that $\lim_{z^*\rightarrow
-\infty} P(y>0|z^*)=0$.

\section{Proof of Proposition \ref{cs_mult}}\label{A222}
\begin{proof}
Pick an arbitrary $t\in[0,1]$ and define
\begin{equation}\label{Attheta}
A_t(\theta) := \left\{\begin{array}{ll}
1 \qquad\qquad & \theta < t \\
0 \qquad\qquad & \text{otherwise.}
\end{array} \right.
\end{equation}
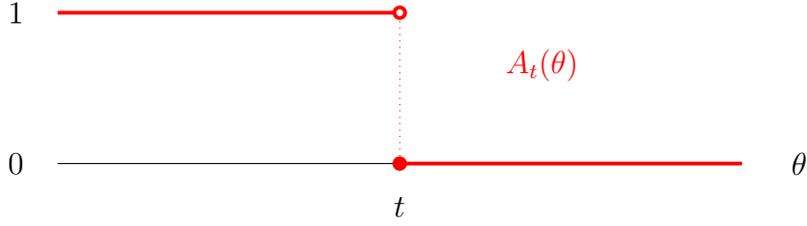
\begin{figure}
\begin{center}
	\begin{tikzpicture}[xscale=1]
	\draw (0,0) -- (9,0);
	\draw [line width=0.5mm,color=red](0,2) -- (4.45,2);
	\draw [line width=0.5mm,color=red](4.5,2) circle (2pt);
	\draw [line width=0.5mm,color=red](4.55,0) -- (9,0);
	\filldraw [line width=0.5mm,color=red](4.5,0) circle (2pt);
	\draw [dotted, color=red](4.5,.1) -- (4.5, 1.9);
	\node [left] at (10,0) {$\theta$};
	\node[align=center, below] at (4.5,-.3){$t$};
	\node[align=center, left] at (-.3,0){$0$};
	\node[align=center, left] at (-.3,2){$1$};
	\node[align=center, left,color=red] at (7,1.3){$A_t(\theta)$};
	\end{tikzpicture}
\end{center}
\caption{Perfect coordination} \label{Diagram2}
\end{figure}
To establish that $A_t$ functions (\ref{Attheta}) are equilibria, we have to show that \eqref{eq_consistency} holds. First, we determine the distribution of signal realizations $y_i$ for a given $\theta$. There are two cases:
\begin{equation}\label{cutoff}
\begin{array}{ll}
\text{Case } \theta<t: & \qquad y_i = 1 + \epsilon^y_i > 1-\sigma \geq \frac12\\
\text{Case } \theta\geq t: & \qquad y_i = 0 + \epsilon^y_i < \sigma \leq \frac12.
\end{array}
\end{equation}
Notice, that signal realizations $y_i$ do not overlap: when $\theta<t,$ all signal realizations are above $\frac12$. On the contrary, whenever $\theta>t,$ all signal realizations are below $\frac12$.
\begin{figure}
\begin{center}
	\begin{tikzpicture}[xscale=1]
	\draw [<->](0,0) -- (9,0);
	\draw (4.5,-.4) -- (4.5, .2);
	\node [left] at (10,0) {$y_i$};
	\node[align=center, above] at (4.5,.3){$\frac12$};
	\draw[decoration={brace,mirror,raise=5pt},decorate]
	(0.1,0) -- node[below=6pt] {$\theta\geq t$} (4.5,0);
	\draw[decoration={brace,mirror,raise=5pt},decorate]
(4.5,0) -- node[below=6pt] {$\theta<t$} (8.9,0);
	\end{tikzpicture}
\end{center}
\caption{Signal cutoff value}\label{Diagram3}
\end{figure}
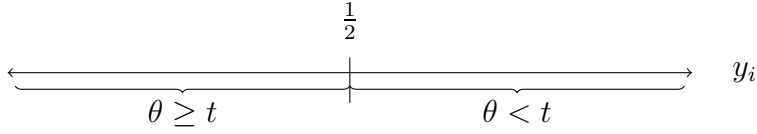
Now suppose that an agents learns that his signal $y_i$ is larger than $\frac12$: then he knows that this is only possible when $\theta<t$. In turn, $\theta<t$ means that a successful attack is underway, which he should join. Indeed, using the conjecture $A_t(\theta)$, we know that all agents will attack as in Figure \ref{Diagram2}.

If the agent's signal $y_i$ is smaller than $\frac12$, he knows $\theta<t$. That is, given the conjectured aggregate attack $A_t(\theta)$, the attack is not successful. Hence, it is optimal for the agent to abstain from attacking.

Equation (\ref{eq_consistency}) holds for all $\theta$: when $\theta<t$ all agents attack, and the aggregate attack is 1. In the other case, no agent attacks and the aggregate attack is 0. Thus, $A_t(\theta)$ satisfies the requirements of an equilibrium.
\end{proof}

\section{Proof of Proposition \ref{infogap_mult_general}}\label{proof_infogap_mult_general}
We begin with a lemma, which collects a number of useful properties. In Section \ref{Iteration} we use these properties, together with the theorem of Arzelà–Ascoli, to show that there exists a continuum of equilibria $A_t(\theta)$. Regarding notation $A^n_t(t^-)$ and $A^n_t(t^+)$  denote left and right limits of the function $A_t(\theta)$ shown in Figure \ref{Diagram3}. Moreover, index $n$ accounts for the n-th iteration over the agents' action $A_t$, and will be used in Section \ref{Iteration}.

\begin{lemma} \label{helper_lemma_2d}
Suppose that for a given $t$ $A(\theta) \in [1-\delta,1]$ if $\theta < t$, and $A(\theta) \in [0,\delta]$ otherwise. Then the following holds:
\begin{enumerate}
\item If $y_i \geq 1-\delta-\gamma$ and $x_i\leq t+ \xi$, then $P[\theta<t|A(\cdot),(x_i,y_i)] \geq c$.
\item If $y_i \leq \delta+\gamma$ and $x\geq t- \xi$, then $P[\theta\geq t|A(\cdot),(x_i,y_i)] \geq c$.
\item $P[y_i\geq 1-\delta-\gamma \cap x\leq t+ \xi|A(\cdot),\theta < t]\geq 1-\delta$
\item $P[y_i\leq \delta+\gamma \cap x> t- \xi|A(\cdot),\theta \geq t]\geq 1-\delta$
\end{enumerate}
\end{lemma}
\begin{proof}
\begin{enumerate}
\item To prove the first statement, we define
$$\kappa_N := \frac{\int_{t}^{t+N} f_x(x_i-\theta)f_y(y_i-A(\theta))\, d\theta}{\int_{t-N}^t f_x(x_i-\theta)f_y(y_i-A(\theta))\, d\theta}\, ,$$
and show that the first two conditions can be reduced to statements about $\kappa_N$. We denote by $\mathcal{U}^t_N$ the uniform distribution of $\theta$ over the interval $[t-N,t+N]$.
\begin{multline*}
P[\theta<t|A(\cdot),(x_i,y_i)] =\lim_{N\rightarrow\infty} \frac{P[\theta<t \wedge (x_i,y_i)|A(\cdot),\mathcal{U}^t_N]}{P[(x_i,y_i)|A(\cdot),\mathcal{U}^t_N]} \\
= \lim_{N\rightarrow\infty} \frac{\int_{t-N}^t \frac1{2N}f_x(x_i-\theta)F_y(y_i-A(\theta)) \, d\theta}{\int_{t-N}^t \frac1{2N}f_x(x_i-\theta)f_y(y_i-A(\theta)) \, d\theta+\int_t^{t+N} \frac1{2N}f_x(x_i-\theta)f_y(y_i-A(\theta)) \, d\theta} \\
= \lim_{N\rightarrow\infty} \frac{1}{1+\kappa_N} \geq c \Leftrightarrow \lim_{N\rightarrow\infty} \kappa_N \leq \frac{1-c}{c}
\end{multline*}
Using appropriate variable transformations ($\tau = \theta - t$ in the nominator and $\tau = t-\theta$ in the denominator) and symmetry of $f_x$ we get
$$\kappa_N = \frac{\int_{0}^{N} f_x(x_i-(t+\tau))f_y(y_i-A(t+\tau))\, d\tau}{\int_{0}^{N} f_x(x_i-(t-\tau))f_y(y_i-A(t-\tau))\, d\tau} \, .$$
Recall, that $y_i \geq 1-\delta-\gamma$ and $x_i\leq t+ \xi$ holds in this case.
\begin{multline*}
\kappa_N \leq \frac{\int_{0}^{N} f_x(x_i-(t+\tau))\, d\tau}{\int_{0}^{N} f_x(x_i-(t-\tau))\, d\tau} \, \frac{\sup_{a\in[0,\delta]} f_y(y_i-a)}{\inf_{a\in[1-\delta,1]} f_y(y_i-a)} \\
= \frac{F_x(x_i-t)-F_x(x_i-t-N)}{F_x(x_i-t+N)-F_x(x_i-t)} \, \frac{\sup_{a\in[0,\delta]} f_y(y_i-a)}{\inf_{a\in[1-\delta,1]} f_y(y_i-a)} \\
\underset{{N\rightarrow \infty}}\longrightarrow \frac{F_x(x_i-t)}{1-F_x(x_i-t)} \, \frac{\sup_{a\in[0,\delta]} f_y(y_i-a)}{\inf_{a\in[1-\delta,1]} f_y(y_i-a)} \\
\leq \frac{F_x(\xi)}{1-F_x(\xi)} \, \frac{\sup_{a\in[0,\delta]} f_y(y_i-a)}{\inf_{a\in[1-\delta,1]} f_y(y_i-a)}\leq \frac{1-c}{c}
\end{multline*}
The last but one inequality uses condition \eqref{cond1_2d}.
\item The proof of the second statement relies on the same arguments used in 1. First, observe that
$$P[\theta\geq t|A(\cdot),(x_i,y_i)] = \lim_{N\rightarrow\infty} \frac{1}{1+\frac1{\kappa_N}} \geq c \Leftrightarrow \lim_{N\rightarrow\infty} \frac1{\kappa_N} \leq \frac{1-c}{c}\, .$$
Note that $y_i \leq \delta+\gamma$ and $x\geq t- \xi$ holds. Again, we can obtain a statement: 
\begin{multline*}
\frac1{\kappa_N} \leq \frac{1-F_x(-\xi)}{F_x(-\xi)}\, \frac{\sup_{a\in[1-\delta,1]} f_y(y_i-a)}{\inf_{a\in[0,\delta]} f_y(y_i-a)} \\
= \frac{F_x(\xi)}{1-F_x(\xi)}\, \frac{\sup_{a\in[0,\delta]} f_y((1-y_i)-a)}{\inf_{a\in[1-\delta,1]} f_y((1-y_i)-a)} \leq \frac{1-c}{c} \, .
\end{multline*}
The equality uses the symmetry of $F_x$ and $f_y$. The last inequality exploits condition \eqref{cond1_2d} (note that $(1-y_i) > 1 - \delta - \gamma$).
\item The third statement follows from condition \eqref{cond2_2d}. Regarding notation, we use $A^n_t(t^-)$ to denote left limits of the function $A_t$ depicted in Figure \ref{Diagram3}:
\begin{multline}
P[y_i\geq 1-\delta-\gamma \wedge x\leq t+ \xi|A(\cdot),\theta < t] 
\\
\geq P[y_i\geq 1-\delta-\gamma \wedge x\leq t+ \xi|A(\theta)=1-\delta,\theta = t^{-}] \\
\geq F_x(t+\xi-t)(1-F_y(1-\delta-\gamma-(1-\delta))) = F_x(\xi)F_y(\gamma)
\geq 1-\delta
\end{multline} 
\item The fourth statement follows from condition \eqref{cond2_2d} and the following inequalities:
\begin{multline}
P[y_i\leq \delta+\gamma \wedge x> t- \xi|A(\cdot),\theta \geq t] \\
\geq P[y_i\leq \delta+\gamma \wedge x> t- \xi|A(\theta)=\delta,\theta = t] \\
\geq F_y(\gamma)(1-F_x(-\xi)) = F_x(\xi)F_y(\gamma)
\geq 1-\delta
\end{multline}
\end{enumerate}
\end{proof}

\subsection{Iteration}\label{Iteration}

We now use an iteration, in which agents' action $A^{n+1}$ is a function of agents' previous action $A^{n}$, to prove the existence of an equilibrium for a given\footnote{The requirement $1-\delta \geq \gamma$, $1>3\delta + 2\gamma$ in Proposition \ref{infogap_mult_general} implies that the interval is non-degenerated.}
\begin{equation}
t\in[\delta+\gamma,1-\delta-\gamma]\, . \label{def_t}
\end{equation}

We start from a hypothetical situation in which player $i$ faces an aggregate attack $A^0_t$ defined by (see green line below)
\begin{equation}\label{Attheta0}
A^0_t(\theta) := \left\{\begin{array}{ll}
1-\delta \qquad\qquad & \theta < t \\
\delta \qquad\qquad & \text{otherwise.}
\end{array} \right.
\end{equation}

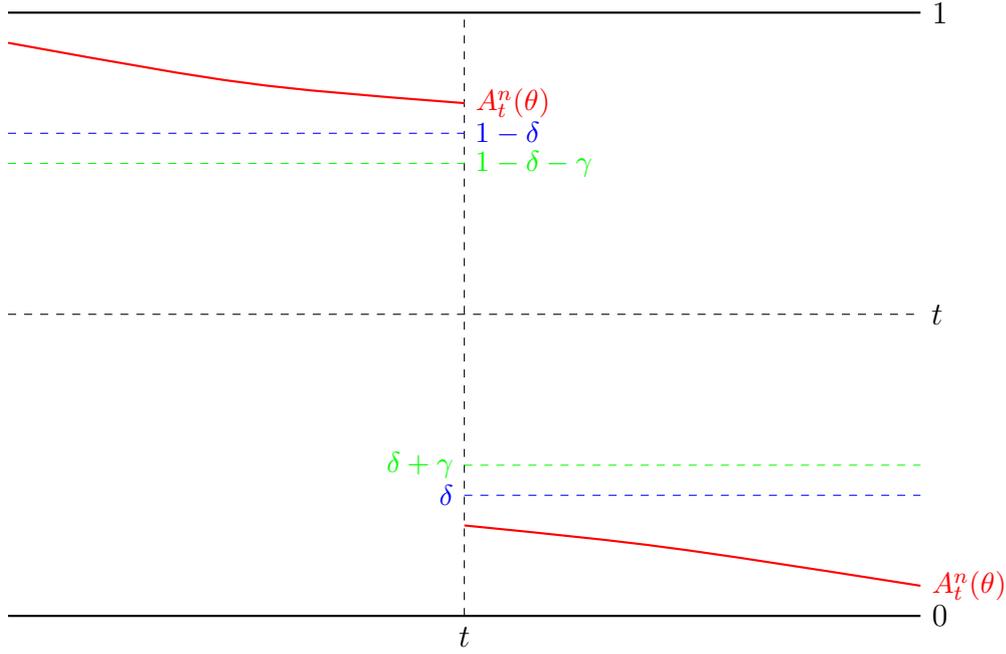
\begin{figure}
\begin{tikzpicture}
\tikzmath{
	\width = 12;
	\height = 8;
	\midheight = \height / 2;
	\atzero = \width / 2;
	\delt = \height / 5;
	\onemindelt = \height - \delt;
	\deltgamma = \height / 4;
	\onemindeltgamma = \height - \deltgamma;
	\upperleft = 0.95 * \height;
	\uppermid = \onemindelt + 0.05 * \height;
	\lowermid = \delt - 0.05 * \height;
	\lowerright = 0.05 * \height;
	\controllefth = (\upperleft + \uppermid) / 2.05;
    \controlleftw = \width / 3.8;
	\controlrighth = (\lowerright + \lowermid) / 1.7;
	\controlrightw = \width * 3 / 4.2;
}
\draw[dashed] (\atzero,0) node[anchor=north] {$t$} -- (\atzero,\height);
\draw[thick] (0,\height) -- (\width,\height) node[anchor=west] {$1$};
\draw[dashed] (0,\midheight) -- (\width,\midheight) node[anchor=west] {$t$};
\draw[dashed, color=blue] (\atzero,\delt) node[anchor=east] {\small $\delta$} -- (\width,\delt) ;
\draw[dashed, color=green] (\atzero,\deltgamma) node[anchor=east] {\small $\delta + \gamma$} -- (\width,\deltgamma) ;
\draw[dashed, color=blue] (0,\onemindelt) -- (\atzero,\onemindelt) node[anchor=west] {\small $1-\delta$};
\draw[dashed, color=green] (0,\onemindeltgamma) -- (\atzero,\onemindeltgamma) node[anchor=west] {\small $1-\delta - \gamma$};
\draw[thick] (0,0) -- (\width,0) node[anchor=west] {$0$};
\draw[thick, color=red] (0,\upperleft) .. controls (\controlleftw,\controllefth) .. (\atzero,\uppermid) node[anchor=west] {\small $A_t^n(\theta)$};
\draw[thick, color=red] (\atzero,\lowermid) .. controls (\controlrightw,\controlrighth) .. (\width,\lowerright) node[anchor=west] {\small $A_t^n(\theta)$};
\end{tikzpicture}
\caption{Construction of $A(\theta)$ via iteration.}\label{Iterationsketch}
\end{figure}

We define the set of signals for which a player who would find it optimal to attack given a conjecture about other players' behavior $A^n_t$ by $\Gamma^n$. Using this definition we have
\begin{equation}
A^{n+1}_t(\theta) := \iint_{\Gamma^n} f_x(x-\theta)f_y(y-A_t^n(\theta))\,dx \, dy\,\label{xyz}
\end{equation} 

By induction it follows from the definition of $A^{n+1}_t$ and the inequalities in Lemma \ref{helper_lemma_2d} that $A^n_t \geq 1-\delta$ ($\leq \delta$) for $\theta\leq t$ (otherwise) for all $n\geq 0$.

To complete the proof of Proposition \ref{infogap_mult_general} we use the theorem of Arzelà–Ascoli to show that the sequence $A_t^n$ has a convergent subsequence. That is, we note that $A_t^n$ are continuous, uniformly bounded (by zero and one), and have a uniformly bounded derivative for all $\theta \neq t$, which implies equicontinuity. Hence, the preconditions of the Arzelà–Ascoli theorem are met\footnote{See, e.g., \citeauthor{shilov2013elementary} (\citeyear{shilov2013elementary}, p. 32).} on each interval $[-k,t]$ (define $A_t^n$ at $t$ by the right limit) and $[t,k+1]$ for $k\in\mathbf{N}$ (they are also met for subsequences). Hence, for $k=1$ there is a convergent subsequence on $[-1,t]$ denoted by $A_t^{n_1}$, from which we can select yet another convergent subsequence on $[t,2]$ (and thus on $[-1,2]$) denoted by $A_t^{n_2}$. We can carry out the same procedure for each $k>1$, and receive a (sub)sequence $A_t^{n_{2k}}$ that converges on $[-k,k+1]$. Last but not least, select the $k$-th element of sequence $A_t^{n_{2k}}$ to create a new sequence $A_t^{n_0}$. This sequence is a subsequence of a converging sequence, and hence it converges.

We denote the limit of this sequence by $A_t$. It constitutes an equilibrium due to the continuity of the best-response operator.


\section{One Private Signal with General Error Terms}\label{Tipping}


In this section we extend our results from Section \ref{Normal}, and show that "precise private information" ensures multiple equilibria for more general distribution functions.

Suppose agents receive signals:
\begin{equation}\label{infogap}
z_i = A-\theta + \rho_i,
\end{equation}
which inform about the attack's net size. 



We denote by $G$ the cumulative distribution function of the error $\rho_i$, and by $g=G'$ the respective density function\footnote{We assume that $G$ is differentiable, i.e., we rule out atoms}. We simplify our proof and assume that $g$ is supposed to be symmetric around $0$.

\begin{prop}\label{infogap_mult_1d}
Suppose that there exist $\delta>0$ and $\gamma>0$ such that the following conditions hold:
\begin{align}
&\frac{1-G(\xi-\alpha)}{G(\xi-\beta)} \geq \frac{1-c}{c}, \quad\text{for all }\xi\geq 1-\delta-\gamma,\, \alpha\in[0,\delta],\text{ and } \beta\in[1-\delta,1] \label{cond1_infogap_mult_1d}\\
& G(\gamma) \geq 1 - \delta \text{, and} \label{cond2_infogap_mult_1d}\\
& g(\delta-\gamma) < 1\label{cond3_infogap_mult_1d}
\end{align}
Then, there exists a continuum of equilibria.
\end{prop}
\begin{proof}
See Appendix \ref{proof_infogap_mult_1d}.
\end{proof}

The proof's strategy is to fix a $t$, and to start with a guess of the equilibrium aggregate attack. In turn, we compute the best response to that attack. This yields another guess for the aggregate attack, to which we compute the best response.... . We use this argument to construct a converging sequence of aggregate attacks. The limit aggregate attack is an equilibrium. This process can be repeated for an open interval of $t$ values, which establishes that there exists a continuum of equilibria.


\section{Proof of Proposition \ref{infogap_mult_1d}}\label{proof_infogap_mult_1d}

The proof is an adapted version of the proof of Proposition \ref{infogap_mult_general}. We begin by proving a supporting lemma:

\begin{lemma} \label{helper_lemma_1d}
Suppose that for a given $t$ $A(\theta) \in [1-\delta,1]$ if $\theta < t$, and $A(\theta) \in [0,\delta]$ otherwise. Then the following holds:
\begin{enumerate}
\item If $z_i \geq 1 - \delta - t- \gamma$ then $P[\theta<t|A(\cdot),z_i] \geq c$.
\item If $z_i \leq \delta - t + \gamma$, then $P[\theta\geq t|A(\cdot),z_i] \geq c$.
\item $P[z_i \geq 1 - \delta - t- \gamma|A(\cdot),\theta < t]\geq 1-\delta$
\item $P[z_i \leq \delta - t + \gamma|A(\cdot),\theta \geq t]\geq 1-\delta$
\end{enumerate}
\end{lemma}
\begin{proof}
\begin{enumerate}
\item To prove the first inequality in 1. we define
$$\kappa_N := \frac{\int_{t}^{t+N} g(z_i-A(\theta)+\theta) \, d\theta}{\int_{t-N}^{t} g(z_i-A(\theta)+\theta) \, d\theta}\, ,$$
and show that both inequalities in 1. can be reduced to statements about $\kappa_N$.
\begin{multline*}
P[\theta<t|A(\cdot),z_i] =\lim_{N\rightarrow\infty} \frac{P[\theta<t \wedge z_i|A(\cdot),\mathcal{U}^t_N]}{P[z_i|A(\cdot),\mathcal{U}^t_N]} \\
= \lim_{N\rightarrow\infty} \frac{\int_{t}^{t+N} g(z_i-A(\theta)+\theta) \, d\theta}{\int_{t}^{t+N} g(z_i-A(\theta)+\theta) \, d\theta+ \int_{t}^{t+N} g(z_i-A(\theta)+\theta) \, d\theta} \\
= \lim_{N\rightarrow\infty} \frac{1}{1+\kappa_N} \geq c \Leftrightarrow \lim_{N\rightarrow\infty} \kappa_N \leq \frac{1-c}{c}
\end{multline*}
Due to continuity of $g$ there exist $\alpha_N\in[0,\delta]$ and $\beta_N\in[1-\delta,1]$ such that
\begin{multline*}\kappa_N = \frac{\int_{t}^{t+N} g(z_i-\alpha_N+\theta) \, d\theta}{\int_{t-N}^{t} g(z_i-\beta_N+\theta) \, d\theta} = \frac{G(z-\alpha_N+\theta+N)-G(z-\alpha_N+\theta)}{G(z-\beta_N+\theta)-G(z-\beta_N+\theta-N)} \\
\underset{{N\rightarrow \infty}}\longrightarrow \frac{1-G(z-\alpha+\theta)}{G(z-\beta+\theta)} \leq \frac{1-c}{c} \, ,
\end{multline*}
where $\alpha\in[0,\delta]$ and $\beta\in[1-\delta,1]$ are the limits of the respective sequence. 
The last inequality exploits condition \eqref{cond1_infogap_mult_1d} and implies that $P[\theta<t|A(\cdot),z_i] \geq c$.
\item The proof of property 2. works along the lines of step 1., or alternatively like step 2. in the proof of Lemma \ref{helper_lemma_2d}.
\item To prove the third property, we use the notation $A^n_t(t^-)$ to denote left limits of the function $A_t$ depicted in Figure \ref{Diagram3}:
\begin{multline*}
P[z_i \geq 1 - \delta - t - \gamma|A(\cdot),\theta < t] \geq  P[z_i \geq 1 - \delta - t - \gamma|A(\theta) = 1- \delta,\theta = t^-] \\
= P[\rho_i\geq - \gamma] = 1 - G(-\gamma) = G(\gamma) \geq 1 - \delta
\end{multline*}
The last but one equality holds due $G$ being symmetric, the last equality due to condition \eqref{cond2_infogap_mult_1d}.
\item Part 4. follows from:
\begin{multline*}
P[z_i \leq \delta - t + \gamma|A(\cdot),\theta \geq t] \geq  P[z_i \leq \delta - t + \gamma|A(\theta) = \delta,\theta =t] \\
= P[\rho_i\leq \gamma] = G(\gamma) \geq 1 - \delta
\end{multline*}
Again, the last equality is due to condition \eqref{cond2_infogap_mult_1d}.
\end{enumerate}
\end{proof}

The remaining part of the proof of Proposition \ref{infogap_mult_general} applies directly with two exceptions:
\begin{enumerate}
\item For given $A^n$ there exists a cutoff $z_n$ such that a player attacks iff $z_i<z_n$. This allows to adjust the definition of $A^{n+1}_t$ given $A^n_t$;. 
$$A^{n+1}_t(\theta) := \int_{-\infty}^{z_n} g (z- A^n_t(\theta)+\theta)) \, dz \, .$$
\item To establish equicontinuity of $A^n_t$, we start with another supporting Lemma:
\begin{lemma}\label{xn_lemma}
	Suppose agents consider the aggregate attack (function) to be equal to $A^n_t$ and that $A^n_t(\theta) > t + \delta $ for $\theta<t$, and $A^n_t(\theta) < t-\delta$ otherwise. Then, the agent's attack cutoff is $z^n\in(-\gamma,\gamma)$.
\end{lemma}
\begin{proof}
	An agent receiving a signal $z_i\geq \gamma$ would attack, where we use the notation $A^n_t(t^-)$ to denote left limits and $A^n_t(t^+)$ to denote right limits:
	\begin{multline*}
	P[\theta < t|z_i=z,A^n_t] =\footnotemark \frac{G(z-A^n_t(t^-)+t)}{G(z-A^n_t(t^-)+t) + 1 - {G(z-A^n_t(t^+)+t)}}\\
	\geq\footnotemark \frac{1-G(z+\delta)}{G(z-\delta) + 1 - G(z+\delta)} 
	= P[\theta < t|z_i=z,A^0_t]> P^*
	\end{multline*}
	Hence,\addtocounter{footnote}{-1}\footnotetext{Use analogous computations as in step 2.}\stepcounter{footnote}\footnotetext{The inequality can be shown by rearranging the following inequality: $$\frac{G(z-A^n_t(t^-)+t)}{G(z-\delta)}\geq 1 \geq \frac{1-G(z-A^n_t(t^+)+t)}{1-G(z+\delta)}$$} for the cutoff $z^n<\gamma$ has to hold. On the other hand, $z^n>-\gamma$ holds, as an agent receiving a signal $z_i\leq -\gamma$ would not attack:
	\begin{multline*}
	P[\theta\geq t|z_i=z,A^n_t] =  \frac{1 - G(z-A^n_t(t^+)+t)}{G(z-A^n_t(t^-)+t) + 1 - {G(z-A^n_t(t^+)+t)}}\\
	\geq \frac{1-G(z+\delta)}{G(z-\delta) + 1 - G(z+\delta)} = P[\theta\geq t|z_i=z,A^0_t]> 1-P^*
	\end{multline*}
\end{proof}

Now, note that $A_t^n$ are continuous, uniformly bounded and have a uniformly bounded derivative for all $\theta \neq t$:
{\small\begin{align*}\frac{d A_t^{n+1}}{d \theta}(\theta) &= - g(z^n-A_t^n(\theta)+\theta)\left(\frac{d A_t^n}{d \theta}(\theta)-1\right) \\
	&= -g(z^n-A_t^n(\theta)+\theta)\, \frac{1-g(z^n-A_t^n(\theta)+\theta)^n}{1-g(z^n-A_t^n(\theta)+\theta)}\\
	&\geq -\max\left\{g(\gamma-\delta)\, \frac{1-g(\gamma-\delta)^n}{1-g(\gamma-\delta)},g(-\gamma+\delta)\, \frac{1-g(-\gamma+\delta)^n}{1-g(-\gamma+\delta)} \right\}\\
	&> \frac{-g(\gamma-\delta)}{1-g(\gamma-\delta)}
	\end{align*}}The second equality can be shown by induction using $\frac{d A_t^0}{d \theta}(\theta) = 0$. We use the notation to $A_t^n(t^-)$ for the right limit of $A_t^n$ for $\theta\rightarrow t$. The first inequality holds because $g$ is increasing (decreasing) on $\theta < t$ (otherwise), $z^n\in(-\gamma,\gamma)$ (Lemma \ref{xn_lemma}), and $A_t^n(\theta)\geq t+ \delta$ on $\theta < t$ and $A_t^n(\theta) < t - \delta$ otherwise. The second inequality holds due to the symmetry of $g$ and condition \eqref{cond3_infogap_mult_1d}. Also note, that the second equality implies that the derivative of $A_t^n$ is non-positive. Thus, we have established uniform upper and lower bounds.
\end{enumerate}

\newpage
\addcontentsline{toc}{section}{References}
\markboth{References}{References}
\bibliographystyle{apalike}
\bibliography{References}

\begin{thebibliography}{}

\bibitem[Angeletos et~al., 2007a]{Pav07}
Angeletos, G.-M., Hellwig, C., and Pavan, A. (2007a).
\newblock Dynamic global games of regime change: learning, multiplicity, and
  the timing of attacks.
\newblock {\em {Econometrica}}, 75:711--756.

\bibitem[Angeletos et~al., 2007b]{Pav07b}
Angeletos, G.-M., Hellwig, C., and Pavan, A. (2007b).
\newblock Dynamic global games of regime change: learning, multiplicity, and
  the timing of attacks.
\newblock {\em {Online Appendix}}.

\bibitem[Angeletos and Werning, 2006]{Ang06}
Angeletos, G.-M. and Werning, I. (2006).
\newblock Crises and prices: Information aggregation, multiplicity, and
  volatility.
\newblock {\em {American Economic Review}}, 96(5):1720--1736.

\bibitem[Battigalli and Guaitoli, 1997]{Bat97}
Battigalli, P. and Guaitoli, D. (1997).
\newblock Conjectural equilibria and rationalizability in a game with
  incomplete information.
\newblock In {\em Decisions, Games and Markets}, pages 97--124. eds. Pierpaolo
  Battigalli and Aldo Montesano and Fausto Panunzi, Springer.

\bibitem[Bergemann and Morris, 2015]{Ber16}
Bergemann, D. and Morris, S. (2015).
\newblock Bayes correlated equilibrium and the comparison of information
  structures in games.
\newblock {\em {Cowles Foundation DP 1909}}.

\bibitem[Binmore and Samuelson, 2001]{Bin01}
Binmore, K. and Samuelson, L. (2001).
\newblock Coordinated action in the electronic mail game.
\newblock {\em {Games and Economic Behavior}}, 35:6--30.

\bibitem[Carlsson and van Damme, 1993]{Car93}
Carlsson, H. and van Damme, E. (1993).
\newblock Global games and equilibrium selection.
\newblock {\em {Econometrica}}, 61:989--1018.

\bibitem[Chamley, 1999]{Cha99}
Chamley, C. (1999).
\newblock Coordinating regime switches.
\newblock {\em {Quarterly Journal of Economics}}, 114(3):869--905.

\bibitem[Esponda, 2013]{Esp13}
Esponda, I. (2013).
\newblock Rationalizable conjectural equilibrium: A framework for robust
  predictions.
\newblock {\em {Theoretical Economics}}, 8:467--501.

\bibitem[Frankel, 2012]{Fra12}
Frankel, D. (2012).
\newblock Recurrent crises in global games.
\newblock {\em {Journal of Mathematical Economics}}, 48:309--321.

\bibitem[Grafenhofer and Kuhle, 2016]{Gra16}
Grafenhofer, D. and Kuhle, W. (2016).
\newblock Observing each other's observations in a bayesian coordination game.
\newblock {\em {Journal of Mathematical Economics}}, 67:10--17.

\bibitem[Hahn, 1977]{Hah77}
Hahn, F. (1977).
\newblock Exercises in conjectural equilibria.
\newblock {\em {Scandinavian Journal of Economics}}, 79(2):210--226.

\bibitem[Izmalkov and Yildiz, 2010]{Izm10}
Izmalkov, S. and Yildiz, M. (2010).
\newblock Investor sentiments.
\newblock {\em {American Economic Journal: Microeconomics}}, 2(1):21--38.

\bibitem[Kovac and Steiner, 2013]{Kov13}
Kovac, E. and Steiner, J. (2013).
\newblock Reversability in dynamic coordination problems.
\newblock {\em {Games and Economic Behavior}}, 77(1):298--320.

\bibitem[Kuhle, 2016]{Kuh15}
Kuhle, W. (2016).
\newblock A global game with heterogenous priors.
\newblock {\em {Economic Theory Bulletin}}, 4(2):167--185.

\bibitem[Loeper et~al., 2014]{Loe14}
Loeper, A., Steiner, J., and Stewart, C. (2014).
\newblock Influential opinion leaders.
\newblock {\em {Economic Journal}}, 124:1147--1167.

\bibitem[Minelli and Polemarchakis, 2003]{Min03}
Minelli, E. and Polemarchakis, H. (2003).
\newblock Information at equilibrium.
\newblock {\em {Economic Theory}}, 21(2-3):573--584.

\bibitem[Morris and Shin, 1998]{Mor98}
Morris, S. and Shin, H.~S. (1998).
\newblock Unique equilibrium in a model of self-fullfilling currency attacks.
\newblock {\em {American Economic Review}}, 88(3):578--597.

\bibitem[Morris and Shin, 2004]{Mor04}
Morris, S. and Shin, H.~S. (2004).
\newblock Coordination risk and the price of debt.
\newblock {\em {European Economic Review}}, 48(1):133--153.

\bibitem[Rubinstein, 1989]{Rub89}
Rubinstein, A. (1989).
\newblock The electronic mail game: Strategic behaviour under 'almost common
  knowledge'.
\newblock {\em {American Economic Review}}, 79(3):385--391.

\bibitem[Rubinstein and Wolinsky, 1994]{Rub94}
Rubinstein, A. and Wolinsky, A. (1994).
\newblock Rationalizable conjectural equilibrium: Between nash and
  rationalizability.
\newblock {\em {Games and Economic Behavior}}, 6(2):299--311.

\bibitem[Shilov, 2013]{shilov2013elementary}
Shilov, G.~E. (2013).
\newblock {\em Elementary functional analysis}.
\newblock Courier Corporation (Reprint of the Prentice-Hall, 1961 edition).

\bibitem[Steiner and Steward, 2011]{Ste11}
Steiner, J. and Steward, C. (2011).
\newblock Communication, timing, and common learning.
\newblock {\em {Journal of Economic Theory}}, 146:230--247.

\end{thebibliography}

\end{document}